\documentclass[a4paper,UKenglish,cleveref, autoref, thm-restate]{lipics-v2021}

\usepackage{tablefootnote}
\usepackage[ruled]{algorithm2e}

\title{Concise and Efficient Quantum Algorithms for Distribution Closeness Testing \footnote{The authors are ordered alphabetically.}}

\author{Lvzhou Li}{Institute of Quantum Computing and Computer Theory, School of Computer and Engineering, Sun Yat-sen University, China}{lilvzh@mail.sysu.edu.cn}{[orcid]}{[funding]}

\author{Jingquan Luo}{Institute of Quantum Computing and Computer Theory, School of Computer and Engineering, Sun Yat-sen University, China}{luojq25@mail2.sysu.edu.cn}{[orcid]}{[funding]}

\authorrunning{Lvzhou Li and Jingquan Luo} 

\Copyright{Lvzhou Li and Jingquan Luo} 

\ccsdesc[300]{Theory of computation~Quantum query complexity}

\keywords{Quantum algorithm, property testing, distribution property} 

\nolinenumbers 

\begin{document}

\maketitle

\begin{abstract}
    We study the impact of quantum computation on the fundamental problem of testing the property of distributions. In particular, we focus on  testing whether two unknown classical distributions are close or far enough, and  propose the  currently  best quantum algorithms for this problem under the metrics of $l^1$-distance and $l^2$-distance. Compared with the latest results given in \cite{gilyen2019distributional} which relied on the technique of quantum singular value transformation (QSVT), our algorithms not only have lower complexity,  but also are more concise.
\end{abstract}

\section{Introduction}

Property testing is a fundamental problem in theoretical computer science, with  the goal being to determine whether the target object has a certain property, under the promise that the object either has the property or is ``far'' from having that property. Quantum computing has a positive impact on many problems, which can improve the complexity of the problems, and property testing problems are no exception. There have been a lot of works on the topic of ``quantum property testing'', and the readers can refer to \cite{montanaro2013survey}. 


A fundamental problem in statistics and learning theory is to test properties of distributions. A few works have shown that quantum computing can speed up testing  properties of distributions, for example, \cite{bravyi2011quantum, montanaro2015quantum, gilyen2019distributional}. In this article, we focus on quantum algorithms for testing whether two unknown distributions are close or far enough, under the metric of $l^1$  or $l^2$ distance. 

\subsection{Query-access models and problem statements}

We firstly give the  definitions of classical and quantum access models for distributions, and then formally describe the problem we consider in this article. 

We consider two probability distributions on $[n]$, and usually denote them by $p$ and $q$.  The $l^\alpha$-distance is adopted to measure the distance between two distributions, which is defined as $|| p - q ||_\alpha := (\sum_{i = 1}^n | p_i - q_i ||^\alpha)^{1/\alpha}$. 

To get access to a classical distribution, the most natural model is sampling.

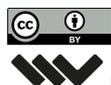
\begin{definition}[Sampling]
    A classical distribution $(p_i)^n_{i=1}$ is accessible via classical sampling if we can request samples from the distribution, i.e., get a random $i \in [n]$ with probability $p_i$.  
\end{definition}

There have been several quantum query-access models for classical distributions. We formulate the most general and natural one as follows, which has been studied in \cite{montanaro2015quantum, gilyen2019distributional, hamoudi2018quantum} and is adopted  throughout this paper.

\begin{definition}[Purified quantum query-access]\label{purified quantum query-access}
A classical distribution $(p_i)^n_{i=1}$ has purified quantum query-access if we have access to a unitary oracle $U_p$ (and its inverse and  controlled versions\footnote{The assumption of having access to the controlled version of $U_p$ is actually also made in \cite{montanaro2015quantum, gilyen2019distributional}.}) acting as \footnote{Without losing generality, we may assume that $B = C^n$, i.e., $B = span\{|1\rangle, \cdots, |n\rangle\}$.}
    \begin{equation*}
        U_p | 0 \rangle_A | 0 \rangle_B = \sum_{i=1}^n \sqrt{p_i} | \phi_i \rangle_A | i \rangle_B 
    \end{equation*}
such that $\langle \phi_i | \phi_j \rangle = \delta_{ij}$.
\end{definition}

Another quantum query-access model for classical distributions is to encode the distribution as a long string and denote probabilities by the frequencies\cite{bravyi2011quantum, chakraborty2010new, li2018quantum, montanaro2015quantum2}. One more model is called pure-state preparation access model\cite{arunachalam2018optimal}, which is more powerful than the model defined in Definition \ref{purified quantum query-access} and assumes that we have access to a unitary $U_p$ acting as $U_p |0\rangle = \sum_{i=0}^n\sqrt{p_i}|i\rangle$.
For more comparison of different quantum access models for classical distributions, one can refer to \cite{gilyen2019distributional, belovs2019quantum}.

In this article, we consider  quantum algorithms for the problem of deciding whether two distributions are close or far enough, which is formally described as follows:
\begin{definition}[{$l^\alpha$}-closeness testing]
    Given $\epsilon > 0$ and access to oracles generating two probability distributions $p$, $q$ on $[n]$, and promised that $p = q$ or $|| p - q ||_{\alpha} \geq \epsilon$, $l^\alpha$-closeness testing requires to  decide which the case is. The problem is called robust testing when the promise is $|| p - q ||_{\alpha} \leq (1-\nu)\epsilon$ or $|| p - q ||_{\alpha} \geq \epsilon$ for $\nu \in (0, 1] $ and one is asked to decide which the case is.
\end{definition}

\subsection{Contributions}
In this article, we give the the currently best quantum algorithms for the above closeness testing problems, i.e. Theorem \ref{theorem1} and its corollaries.
The  comparison of our results with previous classical and quantum results is presented in Table \ref{table1}. For the $l^2$-closeness testing problem, we propose a quantum algorithm (Corollary \ref{corollary1}) with query complexity $O(\frac{1}{\epsilon})$, improving the previous best complexity $O(\frac{1}{\epsilon}\log^3(\frac{1}{\epsilon})\log\log(\frac{1}{\epsilon}))$ given by \cite{gilyen2019distributional}. It is worthy  noting that our algorithm not only drops the polylog factors, but also is more concise, not relying on the QSVT\cite{gilyen2019quantum} frame. For the $l^1$-closeness testing problem, our quantum algorithm (Corollary \ref{corollary2})  is more efficient than all the existing quantum algorithms.

\begin{theorem}\label{theorem1}
    Given purified quantum query-access oracle $U_p$, $U_q$ for two classical distributions $p$, $q$ as in Definition \ref{purified quantum query-access}, for $\nu \in (0, 1]$ and $\epsilon \in (0, 1)$, Algorithm \ref{algorithm1} can decide whether $|| p - q ||_2 \leq (1-\nu)\epsilon $ or $|| p - q ||_2 \geq \epsilon$, with probability at least $\frac{2}{3}$, using $O(\frac{1}{\nu\epsilon})$ calls to $U_p$, $U_q$, their inverse or their controlled version.
\end{theorem}

\begin{corollary}\label{corollary1}
    Given purified quantum query-access oracle $U_p$, $U_q$ for two classical distributions $p$, $q$ as in Definition \ref{purified quantum query-access}, for $\epsilon \in (0, 1)$, there is a quantum tester that can decide whether  $ p = q $ or $|| p - q ||_2 \geq \epsilon$, with probability  at least $\frac{2}{3}$, using $O(\frac{1}{\epsilon})$ calls to $U_p$, $U_q$, their inverse or their controlled version.
\end{corollary}

\begin{proof}
    It directly follows from Theorem \ref{theorem1} with $\nu \leftarrow 1$.
\end{proof}

\begin{corollary}\label{corollary2}
    Given purified quantum query-access oracle $U_p$, $U_q$ for two classical distributions $p$, $q$ as in Definition \ref{purified quantum query-access}, for $\epsilon \in (0, 1)$, there is a quantum tester that can decide whether  $ p = q $ or $|| p - q ||_1 \geq \epsilon$, with probability at least $\frac{2}{3}$, using $O(\frac{\sqrt{n}}{\epsilon})$ calls to $U_p$, $U_q$, their inverse or their controlled version.
\end{corollary}

\begin{proof}
    It directly follows from the Cauchy-Schwartz inequality $|| p - q ||_2 \geq \frac{1}{\sqrt{n}} || p 
 - q ||_1$ and taking $\epsilon \leftarrow \frac{\epsilon}{\sqrt{n}}$ in Theorem \ref{theorem1}.
\end{proof}

\begin{table}[htb]
\renewcommand\arraystretch{1.5} 
\setlength{\belowcaptionskip}{10pt}
\begin{center}   
\caption{ Summary of sample and query complexity results of $l^\alpha$-closeness testing problems.}
\label{table1} 
\begin{tabular}{|c|c|c|c|}   
\hline   &$l^1$-closeness testing & $l^2$-closeness testing \\ 
\hline   Classical sampling & $ \Theta(\max(\frac{n^{2/3}}{\epsilon^{4/3}}, \frac{n^{1/2}}{\epsilon^2})) $ \cite{chan2014optimal} & $\Theta(\frac{1}{\epsilon^2})$ \cite{chan2014optimal} \\
\hline   \cite{bravyi2011quantum} \tablefootnote{\cite{bravyi2011quantum, montanaro2015quantum} actually considered the problem of estimating the total variation distance between probability distributions, which is harder than $l^1$-closeness testing.}  & $ O(\frac{\sqrt{n}}{\epsilon^8}) $ & $/$ \\
\hline   \cite{montanaro2015quantum} & $ O(\frac{\sqrt{n}}{\epsilon^{2.5}}\log(\frac{1}{\epsilon})) $ & $/$ \\
\hline   \cite{gilyen2019distributional} & $ O(\frac{\sqrt{n}}{\epsilon}\log^3(\frac{\sqrt{n}}{\epsilon})\log\log(\frac{\sqrt{n}}{\epsilon})) $ & $O(\frac{1}{\epsilon}\log^3(\frac{1}{\epsilon})\log\log(\frac{1}{\epsilon}))$ \\
\hline   ours & $ O(\frac{\sqrt{n}}{\epsilon})$ & $O(\frac{1}{\epsilon})$ \\
\hline 
\end{tabular}   
\end{center}   
\end{table}

\subsection{Techniques
}
We firstly encode the $l^2$-distance between two distributions $p$ and $q$ into the amplitude of some quantum state and then apply the well-known quantum algorithm \textbf{Amplitude estimation} \cite{brassard2002quantum}. Different from \cite{gilyen2019distributional} which relied on the technique of quantum singular value transformation (QSVT), our method depends on the clever unitary design, and has a more concise form.

\subsection{Related works on distribution  testing}

The topic of reducing the complexity of property testing problems with the help of quantum computing has attracted a lot of research, and refer to \cite{montanaro2013survey} for more details. Here we briefly describe some results on distribution  testing. The first work was due to \cite{bravyi2011quantum} that considered three different problems including $l^1$-closeness testing, uniformity testing and orthogonality testing, and gave quantum upper bounds $O(\sqrt{n}/\epsilon^8)$, $\tilde{O}(n^{1/3})$, $O(n^{1/3}/\epsilon)$, respectively. \cite{chakraborty2010new} independently gave a upper bound $\tilde{O}(n^{1/3}/\epsilon^2)$ for uniformity testing and further showed that testing equality to any fixed distribution can be done with query complexity $\tilde{O}(n^{1/3}/\epsilon^5)$. \cite{montanaro2015quantum} improved the $\epsilon$-dependence of $l^1$-closeness testing to $\tilde{O}(\sqrt{n}/\epsilon^{2.5})$, and the upper bound is further improved to $\tilde{O}(\sqrt{n}/\epsilon)$ by \cite{gilyen2019distributional},where it also showed a upper bound $\tilde{O}(1/\epsilon)$ for $l^2$-closeness testing and a upper bound $\tilde{O}(\sqrt{nm}/\epsilon)$ for independence testing. Both \cite{bravyi2011quantum} and \cite{chakraborty2010new} showed that $\Omega(m^{1/3})$ quantum queries are also necessary for uniformity testing, following from a reduction from the collision problem\cite{aaronson2004quantum, ambainis2005polynomial, kutin2005quantum}.


\section{Preliminaries}

We use $[n]$ to denote the set $\{1, 2, \cdots, n\}$. For a  quantum system composed of $A$, $B$ and a unitary operator $U$, we use $(U)_A$ to denote that the unitary $U$ is operated on the subsystem $A$, while the other subsystem remains unchanged.

We would need a well-known quantum algorithm to estimate some probability $p$ quadratically more efficiently than  classical sampling:

\begin{theorem}[Amplitude estimation\cite{brassard2002quantum}]\label{amplitude estimation}

    Given a unitary $U$ and an orthogonal projector $\Pi$ such that $U | 0 \rangle = \sqrt{p} | \phi 
 \rangle + \sqrt{1-p} | \phi^{\perp} \rangle$ for some $p \in [0, 1]$, $\Pi | \phi \rangle = | \phi \rangle$ and $ \Pi | \phi^{\perp} \rangle = 0$, there is a quantum algorithm called \textbf{amplitude estimation} outputing $\tilde{p}$, such that 
    \begin{equation*}
        | \tilde{p} - p | \leq 2\pi \frac{\sqrt{p(1-p)}}{t} + \frac{\pi^2}{t^2}
    \end{equation*}
 with probability at least $8/\pi^2$, using $t$ calls to $U$, $U^\dagger$ and $I - 2\Pi$.
\end{theorem}

\section{Quantum Algorithm}

In this section, we give the proof of Theorem \ref{theorem1}. The formal description of our quantum algorithm is given in Algorithm \ref{algorithm1}.


\begin{algorithm}[htb]
    \SetKwInOut{KWProcedure}{Procedure}
    \SetKwInput{Runtime}{Runtime}
    \caption{Quantum tester for (robust) distribution $l^2$-closeness testing problem with purified quantum query-access.}\label{algorithm1}  
    \LinesNumbered
    \KwIn {Purified quantum query-access oracle for two classical distributions $p$ and $q$ as in Definition \ref{purified quantum query-access}; two real numbers $\nu \in (0, 1]$ and $\epsilon \in (0, 1)$.}
    \KwOut {Output \textbf{CLOSE} when $|| p - q ||_2 \leq (1-\nu)\epsilon $, or output \textbf{FAR} when $|| p - q ||_2 \geq \epsilon $, with probability at least $\frac{2}{3}$.}
    \Runtime{$O(\frac{1}{\nu\epsilon})$.}
    
    \KWProcedure{}
    Prepare the initial state $| \psi_0 \rangle = | 0 \rangle_A | 0 \rangle_B | 0 \rangle_{C} | 0 \rangle_D$, where the system $C$ has the same dimension as the system $B$, and the last system is a qubit (a two-dimensional Hilbert space).

    Perform  $U := (H)_D (\tilde{U}_p \otimes | 0 \rangle \langle 0 | + \tilde{U}_q \otimes | 1 \rangle \langle 1 |)_{ABCD} (HX)_D $ on  $| \psi_0 \rangle$, where $\widetilde{U}_p := (U_p^{\dagger})_{AB} (U_{copy})_{BC} (U_p)_{AB}$ (similar for $\tilde{U}_q$), $U_{copy}$ is a unitary operation mapping  $| i \rangle | 0 \rangle \rightarrow | i \rangle | i \rangle $ for $i = 1, 2, \cdots, n$, $H$ is the Hadamard gate, and $X$ is the Pauli-X gate. 

    Run amplitude estimation in Theorem \ref{amplitude estimation}, with $U$ defined above, $\Pi := | 0 \rangle_A | 0 \rangle_B \langle 0 |_A \langle 0 |_B \otimes I \otimes | 0 \rangle_D \langle 0 |_D$, and $t := \frac{10\pi}{\nu\epsilon}$, and denote by $\Delta '$ the result.

    \eIf{$\Delta ' < (\frac{1}{4} - \frac{\nu}{8})\epsilon^2$}{output \textbf{CLOSE}}{output \textbf{FAR}}
    
\end{algorithm}


To prove the main theorem, we firstly give a key lemma as follows:

\begin{lemma}\label{lemma:1}
Let $U_p$ be a unitary quantum query-access oracle for a classical distribution $p$ as in Definition \ref{purified quantum query-access}, and let $\tilde{U}_p := (U_p^{\dagger} \otimes I ) (I \otimes U_{copy}) (U_p \otimes I)$, and $\Pi := | 0 \rangle \langle 0 | \otimes | 0 \rangle \langle 0 | \otimes I$.  Then we have 
    \begin{equation*}
        \tilde{U}_p |0\rangle | 0 \rangle | 0 \rangle = \sum_{i = 1}^n p_i | 0 \rangle | 0 \rangle | i \rangle + | 0^{\perp} \rangle,
    \end{equation*}
where $ \Pi | 0^{\perp} \rangle = 0$.
\end{lemma}

\begin{proof}
For $k \in \{1, \cdots, n\}$, we have
    \begin{align*}
        \langle 0 | \langle 0 | \langle k | \tilde{U}_p | 0 \rangle | 0 \rangle | 0 \rangle 
        & = \langle 0 | \langle 0 | \langle i | (U_p^{\dagger} \otimes I ) (I \otimes U_{copy}) (U_p \otimes I) |0\rangle | 0 \rangle | 0 \rangle \\
        & = (\sum_{i=1}^n \sqrt{p_i} \langle \phi_i | \langle i | \langle k) | (I \otimes U_{copy}) (\sum_{j=1}^n \sqrt{p_i} | \phi_j \rangle | j \rangle) | 0 \rangle \\  
        & = (\sum_{i=1}^n \sqrt{p_i} \langle \phi_i | \langle i | \langle k) | (\sum_{j=1}^n \sqrt{p_i} | \phi_j \rangle | j \rangle | j \rangle )\\  
        & = p_k.
    \end{align*}
\end{proof}

\begin{proof}[Proof of Theorem \ref{theorem1}]
    Let $U := (H)_D (\tilde{U}_p \otimes | 0 \rangle \langle 0 | + \tilde{U}_q \otimes | 1 \rangle \langle 1 |)_{ABCD} (HX)_D$ and $\Pi := | 0 \rangle_A | 0 \rangle_B \langle 0 |_A \langle 0 |_B \otimes I \otimes | 0 \rangle_D \langle 0 |_D$ as defined in Algorithm \ref{algorithm1}. We have
    \begin{align}
        | \psi \rangle
        & = U | 0 \rangle | 0 \rangle | 0 \rangle | 0 \rangle 
          = (H)_D (\tilde{U}_p \otimes | 0 \rangle \langle 0 | + \tilde{U}_q \otimes | 1 \rangle \langle 1 |)_{ABCD} (HX)_D | 0 \rangle | 0 \rangle | 0 \rangle | 0 \rangle \\
        & = (H)_D (\tilde{U}_p \otimes | 0 \rangle \langle 0 | + \tilde{U}_q \otimes | 1 \rangle \langle 1 |)_{ABCD} | 0 \rangle | 0 \rangle | 0 \rangle (\frac{1}{\sqrt{2}} | 0 \rangle - \frac{1}{\sqrt{2}} | 1 \rangle) \\
        & = (H)_D (\sum_{i=1}^n \frac{p_i}{\sqrt{2}} | 0 \rangle | 0 \rangle | i \rangle | 0 \rangle - \sum_{i=1}^n \frac{q_i}{\sqrt{2}} | 0 \rangle | 0 \rangle | i \rangle | 1 \rangle + | 0^{\perp} \rangle) \\
        & = \sum_{i=1}^n \frac{p_i - q_i}{2} | 0 \rangle | 0 \rangle | i \rangle | 0 \rangle + \sum_{i=1}^n \frac{p_i + q_i}{2} | 0 \rangle | 0 \rangle | i \rangle | 1 \rangle + | 0^{\perp} \rangle,
    \end{align}
    where $( | 0 \rangle \langle 0 | \otimes | 0 \rangle \langle 0 | \otimes I \otimes I) | 0^{\perp} \rangle = 0$.

    Therefore, we have $|| \Pi \psi \rangle ||^2 = \frac{1}{4} \sum_{i = 1}^n (p_i - q_i)^2 = \frac{|| p - q ||_2^2 }{4}$. Let $\Delta := \frac{|| p - q ||_2^2 }{4} $, and $\Delta '$ be the estimate for $\Delta$ computed by amplitude estimation with $t := \frac{20\pi}{\nu \epsilon}$ as in Algorithm \ref{algorithm1}.

    It remains to show that if $|| p - q ||_2 \leq (1-\nu) \epsilon$ we have $\Delta' \leq (\frac{1}{4} - \frac{\nu}{8})\epsilon^2$ with high probability, and if $|| p - q ||_2 \geq \epsilon$, we have $\Delta' > (\frac{1}{4} - \frac{\nu}{8})\epsilon^2$ with high probability. If $|| p - q ||_2 \leq (1-\nu) \epsilon$, then $\Delta \leq \frac{\epsilon^2 - \nu \epsilon^2 }{4}$; by Theorem \ref{amplitude estimation}, with probability at least $\frac{8}{\pi^2}$, the estimate $\Delta '$ satisfies $| \Delta ' - \Delta| \leq \frac{2\pi \sqrt{\Delta}}{t} + \frac{\pi^2}{t^2} < \frac{\nu\epsilon^2}{8}$, and thus $\Delta' \leq (\frac{1}{4} - \frac{\nu}{8})\epsilon^2$. 
    If $|| p - q ||_2 \geq \epsilon$, then $\Delta \geq \frac{ \epsilon^2 }{4}$. We have the following cases:
    \begin{itemize}
        \item $\Delta \leq \epsilon^2$: by Theorem \ref{amplitude estimation}, with probability at least $\frac{8}{\pi^2}$, the estimate $\Delta ' $ satisfies $| \Delta ' - \Delta | \leq \frac{2\pi \epsilon}{t} + \frac{\pi^2}{t^2} < \frac{\nu\epsilon^2}{8}$, and thus $\Delta' > (\frac{1}{4} - \frac{\nu}{8})\epsilon^2$;

        \item $\Delta > \epsilon^2$: by Theorem \ref{amplitude estimation}, with probability at least $\frac{8}{\pi^2}$, the estimate $\Delta ' $ satisfies $| \Delta ' - \Delta | \leq \frac{2\pi \sqrt{\Delta}}{t} + \frac{\pi^2}{t^2} < \frac{\Delta}{2}$, and thus $\Delta ' > \frac{\epsilon^2}{2} > (\frac{1}{4} - \frac{\nu}{8})\epsilon^2$.
    \end{itemize}
    \textbf{Complexity.} The query complexity of constructing $U$ is $O(1)$, and it takes $O(\frac{1}{\nu \epsilon})$ queries to $U$ to run amplitude estimation. Therefore, the overall query complexity of Algorithm \ref{algorithm1} is $O(\frac{1}{\nu \epsilon})$.
\end{proof}



\section{Conclusion and open problems}

In this article, we have proposed   concise and efficient quantum algorithms for testing whether two classical distributions are close or far  enough under the metrics of $l^1$-distance and $l^2$-distance, with  query complexity  $O(\sqrt{n}/\epsilon)$ and $O(1/\epsilon)$, respectively, improving the results in \cite{gilyen2019distributional} which relied on the technique of quantum singular value transformation(QSVT).

In the following we list some open problems for future work.

\begin{itemize}
    \item Can we improve the precision dependence of estimating the total variation distance between probability distributions? The currently best upper bound is $O(\sqrt{n}/\epsilon^{2.5})$ due to \cite{montanaro2015quantum}.

    \item Can we prove an $\Omega(\sqrt{n}/\epsilon)$ lower bound on $l^1$-closeness testing?

    \item Can our Lemma \ref{lemma:1} be applied to other distribution  testing problems?
\end{itemize}

\bibliography{lipics-v2021-sample-article}


\end{document}